\documentclass[]{interact}

\usepackage{epstopdf}
\usepackage[caption=false]{subfig}
\usepackage{multirow}
\usepackage{xcolor}

\usepackage[numbers,sort&compress]{natbib}
\bibpunct[, ]{[}{]}{,}{n}{,}{,}
\renewcommand\bibfont{\fontsize{10}{12}\selectfont}

\theoremstyle{plain}
\newtheorem{theorem}{Theorem}[section]

\newtheorem{proposition}[theorem]{Proposition}
\newcommand{\f}{\operatorname}
\theoremstyle{definition}

\theoremstyle{remark}

\begin{document}

\articletype{ }

\title{The Fr\'echet distribution: Estimation and Application an Overview}

\author{
\name{P.~L. Ramos\textsuperscript{a}\thanks{CONTACT P.~L. Ramos. Email: pedrolramos@usp.br}, Francisco Louzada\textsuperscript{a}, Eduardo Ramos\textsuperscript{a} and Sanku Dey\textsuperscript{b}}
\affil{\textsuperscript{a}Institute of Mathematical and Computer Sciences, USP Sao Carlos, Brazil; \\ \textsuperscript{b}Department of Statistics, St. Anthony's College, Shillong, Meghalaya, India}
}

\maketitle

\begin{abstract}
In this article we consider the problem of estimating the parameters of the Fr\'echet distribution from both frequentist and Bayesian points of view. First we briefly describe different frequentist approaches, namely, maximum likelihood, method of moments, percentile estimators, L-moments, ordinary and weighted least squares, maximum product of spacings, maximum goodness-of-fit estimators  and compare them  with respect to mean relative estimates,  mean squared errors and the 95\% coverage probability of the asymptotic confidence intervals using extensive numerical simulations. Next, we consider the Bayesian inference approach using reference priors. The Metropolis-Hasting algorithm is used to draw Markov Chain Monte Carlo samples, and they have in turn been used to compute the Bayes estimates and also to construct the corresponding credible intervals. Five real data sets related to the minimum flow of water on Piracicaba river in Brazil are used to illustrate the applicability of the discussed procedures. 
\end{abstract}

\begin{keywords}
 Bayesian Inference; Fr\'echet distribution; Hydrological applications; Maximum product of spacings; Reference prior.
\end{keywords}

\section{Introduction}

Extreme value theory plays an important role in statistical analysis. The most used distribution to describe extreme data is the generalized extreme value (GEV) distribution \cite{jenkinson}. Its cumulative density function (CDF) is given by 
\begin{equation}\label{gev}
F(t|\sigma,\mu,\xi)=
\begin{cases}
\exp\left\{ -\left[1+\xi(t-\mu)/\sigma\right]_+^{-1/\xi}\right\}, \mbox{ for} & \xi\neq 0 \\
\exp\left\{ -\exp\left[-(t-\mu)/\sigma\right]\right\}, \mbox{~~~~~~    for} & \xi= 0
\end{cases}
\end{equation}
where $\sigma>0$, $\mu, \xi \in \mathbb{R}$ and $x_+=\max(x,0)$. Gumbel, Weibull and Fr\'echet distributions are special cases of the so-called generalized extreme value (GEV) distribution. The Fr\'echet distribution is named after French mathematician Maurice Ren\'e Fr\'echet, who developed it in the 1920s as a maximum value distribution (which is also known as the extreme value distribution of type II). Kotz and Nadarajah \cite{kotz2000extreme} describe this distribution and discussed its wide applicability in different spheres such as accelerated life testing,  natural calamities, horse racing, rainfall, queues in supermarkets, sea currents, wind speeds, track race records and so on.

Let the random variable T follows Fr\'echet distribution then its probability density function (PDF) and cdf are given by
\begin{equation}\label{fdpiw}
f(t|\lambda,\alpha)=\lambda\alpha t^{-(\alpha+1)}e^{-\lambda t^{-\alpha}} \quad \mbox{and} \quad F(t|\lambda,\alpha)= e^{-\lambda t^{-\alpha}},
\end{equation}
for all $t>0$  and the quantities $\alpha> 0$ and $\lambda > 0$ are the shape and the scale parameters respectively. This distribution is also referred as Inverse Weibull distribution. The PDF can be unimodal or decreasing depending on the choice of the shape parameter while its hazard function is always unimodal. In this respect, the behavior of Fr\'echet distribution and the log-normal distribution is quite similar.

Several researchers have studied different aspects of inferential procedures for the Fr\'echet distribution. From the classical perspective,  Calabria and Pulcini \cite{calabria1989confidence} and Erto \cite{erto1989genesis} discussed the properties of the maximum likelihood estimators (MLE) and the ordinary least-square estimators (LSE) respectively. Ramos et al. \cite{ramos2017long} presented the MLE for the the Fr\'echet distribution in the presence of cure fraction. Loganathan and Uma \cite{loganathan2017comparison} compared the MLE, LSE, weighted LSE (WLSE) and the method of moments (MME) and outlined that the WLSE provided similar results. Salman et al. \cite{salman2003order} and Maswadah \cite{maswadah2003conditional} studied the Fr\'echet distribution in the context of order statistics and generalized order statistics respectively.
The Bayes estimators were discussed by Calabria and Pulcini \cite{calabria1994bayes} and Kundu and Howlader \cite{kundu2010bayesian} using informative or subjective priors such as Gamma priors (also known as flat priors). However, Bernardo \cite{bernardo2005} argued that the use of simple proper flat priors presumed to be non-informative often hide important unwarranted assumptions which may easily dominate, or even invalidate the statistical analysis and should be strongly discouraged. Recently, Abbas and Tang \cite{abbas2015analysis} studied Fr\'echet distribution based on  Jeffreys and reference priors.

Parameter estimation is significant for any probability distribution and therefore  various estimation methods are frequently studied in the statistical literature. Traditional estimation methods such as the MLE, MME, LSE and WLSE are often opted for parameter estimation. Each  has its own merits and demerits but  the most popular method of estimation  is  the maximum likelihood estimation method. Besides the above cited methods, we consider five additional methods to estimate the parameters of  Fr\'echet  distribution. These additional methods are the maximum product spacing estimator (MPS), percentile estimator (PE), Cram\'er-von-Mises estimator (CME), Anderson-Darling estimator (ADE) and L-moment (LME) estimator. Further, from a Bayesian point of view different Bayes estimators are discussed using objective priors and different loss functions. Also, the coverage probability with a confidence level equal to 95\%  for the estimates are obtained. To evaluate the performance of the estimators, a simulation study is carried out. Finally, five real life data sets have been analyzed for illustrative purposes.

The objective of this paper is to estimate the parameters of the model from both frequentist and Bayesian perspective and to develop a guideline for choosing the best estimation method for the Fr\'echet distribution, which we would be of profound interest to applied statisticians. The choice of the methods of estimation vary among the users and area of applications. With computational advances, the need to have an estimator with closed form has decreased substantially. Thus, a user may prefer to employ the uniformly minimum variance estimation
method although the estimator does not have a closed form expression.

The present study is unique because of the fact that thus far, no attempt has been made to compare all these aforementioned estimators for the two-parameter Fr\'echet distribution. At the end, we present the better estimation procedure for the Fr\'echet distribuion. Additionally, we provide the necessary codes in R to perform such inference. In the last decade, several authors have compared different estimation methods for different distributions. Notable among them are: Kundu and Raqab \cite{kundu2005generalized} for generalized Rayleigh distributions ; Alkasasbeh and Raqab \cite{AlRa09} for generalized logistic distributions;  Teimouri et al. \cite{teimouri2013comparison} for the Weibull distribution;  Mazucheli et al. \cite{mazucheli2013comparison} for weighted Lindley distribution; Rodrigues et al. \cite{rodrigues2016poisson} for the Poisson-exponential distribution; Ramos and Louzada \cite{ramoslouzada2016} for the generalized weighted Lindley distribution and  Dey et al. \cite{DeKu14, dey2016two, dey2016exponentiated, dey2017statistical} for two-parameter Rayleigh distribution,  two-parameter Maxwell distribution, exponentiated Chen distribution and transmuted-Rayleigh respectively.
 
The paper is organized as follows. Section 2 describes nine frequentist methods of estimation. The
Bayes estimators are presented in section 3. In Section 4, we present the Monte Carlo simulation results.
In Section 5, the usefulness of the Fr\'echet distribution is illustrated by using five real data sets. Finally,
concluding remarks are provided in Section 6.

\section{Classical parameter estimation methods}

In this section, we describe nine methods for estimating the parameters  $\lambda$ and $\alpha$  of the Fr\'echet distribution.

\subsection{Maximum Likelihood Estimation}
Among the statistical inference methods, the maximum likelihood method is widely used due its desirable properties including consistency, asymptotic efficiency and invariance. Under the maximum likelihood method, the estimators are obtained by maximizing the log-likelihood function. Let $T_1,\ldots,T_n$ be a random sample such that $T\sim$ Fr\'echet$(\lambda,\alpha)$. Then, the likelihood function from (\ref{fdpiw}) is given by
\begin{equation}\label{loglikeiw}
L(\lambda, \alpha| t) = \prod_{i=1}^n f(t_i, \lambda, \alpha ) =\lambda^n\alpha^n \left(\prod_{i=1}^n t_i^{-(\alpha+1)}\right)\exp\left(-\lambda\sum_{i=1}^n t_i^{-\alpha}\right).
\end{equation}
The log-likelihood function (\ref{loglikeiw}) is given by
\begin{equation*}
l(\lambda, \alpha|t)=n\log(\lambda)+n\log(\alpha)-(\alpha+1)\sum_{i=1}^{n}\log(t_i)-\lambda\sum_{i=1}^n t_i^{-\alpha}.
\end{equation*}

From ${\partial}l(\lambda, \alpha|\boldsymbol{t})/{\partial \lambda}=0$ and ${\partial}l(\lambda, \alpha|\boldsymbol{t})/{\partial \alpha}=0$, we get the likelihood equations
\begin{equation}\label{likelihood1}
\frac{n}{\lambda} -\sum_{i=1}^n t_i^{-\alpha}=0  
\end{equation}
and
\begin{equation}\label{likelihood2}
\frac{n}{\alpha} -\sum_{i=1}^n\log(t_i)+\lambda\sum_{i=1}^n t_i^{-\alpha}\log(t_i)=0,
\end{equation}
whose solutions provide  $\hat\lambda_{MLE}$ and $\hat\alpha_{MLE}$. After some algebraic manipulations, the estimate $\hat\alpha_{MLE}$ can be obtained by solving the following non-linear equation
\begin{equation}
\frac{n}{\alpha} -\sum_{i=1}^n\log(t_i)+\frac{n\sum_{i=1}^n t_i^{-\alpha}\log(t_i)}{\sum_{i=1}^n t_i^{-\alpha}}=0.
\end{equation}

The estimate $\hat\lambda_{MLE}$ can be obtained by substituting $\hat{\alpha}_{MLE}$ in  
\begin{equation}
\hat\lambda_{MLE}= \frac{n}{\sum_{i=1}^n t_i^{-\alpha}}.
\end{equation} 

The obtained ML estimates are asymptotically normally distributed with a joint bivariate normal distribution given by
\begin{equation*} (\hat{\lambda}_{MLE},\hat{\alpha}_{MLE})\sim N_2[(\lambda,\alpha),I^{-1}(\lambda,\alpha))] \mbox{ for } n \to \infty , \end{equation*}
where $I(\lambda,\alpha)$ is the Fisher information matrix given by
\begin{equation}\label{mfishergg}
I(\lambda,\alpha)=
\begin{bmatrix}
\dfrac{n}{\lambda^2} & \dfrac{n\left(1-\gamma-\log(\lambda) \right)}{\lambda\alpha} \\
\dfrac{n\left(1-\gamma-\log(\lambda)\right)}{\lambda\alpha}  & \dfrac{n}{\alpha^2}\left(\tfrac{\pi^2}{6} +\left(1-\gamma-\log(\lambda)\right)^2 \right)
\end{bmatrix} ,
\end{equation}
and $\gamma\approx 0.5772156649$ is known as Euler-Mascheroni constant.


In the following we prove the existence and uniqueness of MLEs.

\begin{theorem}\label{themleune} Let $t_1,\cdots, t_n$ be not all equal. Then the
MLEs of the parameters $\alpha $ and $\lambda $ are unique and are given
by $\widehat{\alpha}$ and 
\begin{equation}
\widehat{\lambda }=\frac{n}{\sum\limits_{i=1}^{n}\,t_{i}^{-\widehat{\alpha }}} , 
\end{equation}%
where $\hat{\alpha}$ is the only solution of non-linear equation
\begin{equation*}
G(\alpha )=\frac{n}{\alpha }-\sum_{i=1}^{n}\log t_{i}-\frac{n}{
\sum\limits_{i=1}^{n}t_{i}^{-\alpha }}\sum_{i=1}^{n}t_{i}^{-\alpha }\log t_{i}.
\end{equation*}
\end{theorem}

\begin{proof} See Appendix A.
\end{proof}

\subsection{Moments Estimators}
The method of moments is one of the oldest method used for estimating the parameters of the statistical models. The raw moments of T for the Fr\'echet distribution is
\begin{equation}\label{propiw3}
E(T^r|\lambda, \alpha) = \lambda^{\tfrac{r}{\alpha}}\Gamma\left(1-\frac{r}{\alpha} \right),
\end{equation}
where $r\in N$ and  $\Gamma(\lambda)=\int_{0}^{\infty}{e^{-x}x^{\lambda-1}dx}$ is the gamma function. Note that $E(T^r|\gamma, \alpha)$ does not have a finite value for $\alpha>r$. The moment estimators (MEs) for the Fr\'echet distribution can be obtained by equating the first two theoretical moments  with the sample moments. However, instead of equating the first two theoretical moments, we consider that
\begin{equation}\label{mmoments1}
E(T|\lambda, \alpha) = \lambda^{\tfrac{1}{\alpha}}\Gamma\left(1-\frac{1}{\alpha} \right) \ \ \ \mbox{ and } \ \ \ Var(T|\lambda, \alpha) =\lambda^{\tfrac{2}{\alpha}}\left( \Gamma\left(1-\frac{2}{\alpha} \right)-\Gamma^2\left(1-\frac{1}{\alpha} \right)\right).
\end{equation}
Therefore, the population coefficient of variation is given by
\begin{equation}
CV(X|\lambda, \alpha) = \frac{\sqrt{Var(T|\lambda, \alpha)}}{E(T|\lambda, \alpha)}=\sqrt{\frac{\Gamma\left(1-2\alpha^{-1}\right)}{\Gamma^2\left(1-\alpha^{-1}\right)}-1}, \nonumber
\end{equation}
which is independent of the scale parameter $\lambda$. So, the estimator for $\hat{\alpha}_{ME}$ can be obtained by solving the following non-linear equation
\begin{equation}
\sqrt{\frac{\Gamma\left(1-2\alpha^{-1}\right)}{\Gamma^2\left(1-\alpha^{-1}\right)}-1}-\frac{s}{\bar{t}} = 0.\nonumber
\end{equation} 
Substituting $\hat{\alpha}_{MME}$ in (\ref{mmoments1}) the estimate $\hat\lambda_{MME}$ can be obtained by solving 
\begin{equation}
\hat\lambda_{ME}= \frac{\bar{t}^\alpha}{\Gamma^\alpha\left(1-\alpha^{-1}\right)}.
\end{equation} 

However, this estimator can only be computed for $\alpha>2$ which is undesirable.

\subsection{Percentile Estimator}\label{PE}

The percentile estimator is a statistical method used to estimate the parameters by comparing the sample points with the theoretical points. This method was originally suggested by Kao \cite{kao1958computer, kao1959graphical} and has been widely used for distributions that has the quantile function in closed form expression, such as the Weibull and the Generalized Exponential distribution. The quantile function of the Fr\'echet distribution has the closed form and is given by
\begin{equation}\label{qufuiw}
Q(p|\lambda, \alpha)=\left(\frac{1}{\lambda}\log\left(\frac{1}{p_i}\right)\right)^{-\frac{1}{\alpha}}.
\end{equation}

Therefore, the percentile estimates (PCEs), $\widehat{\lambda}_{PCE}$ and $\widehat{\alpha}_{PCE}$, can be obtained by minimizing
\begin{equation*}
\sum_{i=1}^n \left(t_{(i)}-\left(\frac{1}{\lambda}\log\left(\frac{1}{p_i}\right)\right)^{-\frac{1}{\alpha}}\right)^2,
\end{equation*}
with respect to $\lambda$ and $\alpha$, where $p_i$ denotes an estimate of $F(t_{(i)}; \lambda, \alpha)$ and $t_{(i)}$ is the ith order statistics (we assume the same notation for the next sections). The estimates of $\lambda$ and $\alpha$ can  be also be obtained by solving the following non-linear equations:
\begin{equation}
\sum_{i=1}^n \left[t_{(i)}-\left(\frac{1}{\lambda}\log\left(\frac{1}{p_i}\right)\right)^{-\frac{1}{\alpha}}\right]\left(\frac{1}{\lambda}\log\left(\frac{1}{p_i}\right)\right)^{-\frac{1}{\alpha}}=0,\nonumber
\end{equation}
\begin{equation}
\sum_{i=1}^n \left[t_{(i)}-\left(\frac{1}{\lambda}\log\left(\frac{1}{p_i}\right)\right)^{-\frac{1}{\alpha}}\right]\log\left(\frac{1}{\lambda}\log\left(\frac{1}{p_i}\right)\right)\left(\frac{1}{\lambda}\log\left(\frac{1}{p_i}\right)\right)^{-\frac{1}{\alpha}}=0,\nonumber
\end{equation}
respectively. In this paper, we consider $p_i = \dfrac{i}{n+1}$. However, several estimators of $p_i$ can be used instead \cite{mann1974methods}. 

\subsection{L-Moments Estimators}

Hosking \cite{hosking1990moments} proposed an alternative method of estimation analogous to 
conventional moments, namely L-moments estimators. L-moments estimators can be obtained by equating the sample with the population L-moments. Hosking \cite{hosking1990moments} stated that the L-moment estimators are more robust than the usual moment estimators, and also relatively robust to the effects of outliers and reasonably efficient when compared to the MLE for some distributions.

For the Fr\'echet distribution, the L-moments estimators can be obtained by equating the first two sample L-moments with the corresponding population L-moments. The first two sample L-moments are
\begin{equation*}
l_1 = \frac{1}{n}\sum_{i=1}^n t_{(i)},\qquad l_2 = \frac{2}{n(n-1)}\sum_{i=1}^n (i-1)t_{(i)}-l_1,
\end{equation*}
where $t_{(1)}, t_{(2)},\cdots,t_{(n)}$ denotes the order statistics of a random sample from a distribution function $F(\boldsymbol{t}| \lambda, \alpha)$. The first two population L-moments are
\begin{equation}\label{lmoments1}
\mu_{1}(\lambda, \alpha) =\int_{0}^{1}Q(p|\lambda, \alpha)dp=E(X|\lambda, \alpha) = \lambda^{\tfrac{1}{\alpha}}\Gamma\left(1-\frac{1}{\alpha} \right) \   
\end{equation}
and
\begin{equation}\label{lmoments2}
\mu_{2}(\lambda, \alpha) =\int_{0}^{1}Q(p|\lambda, \alpha)(2p-1)dp=\lambda^{\tfrac{1}{\alpha}}\left(2^{\tfrac{1}{\alpha}}-1 \right)\Gamma\left(1-\frac{1}{\alpha} \right) ,
\end{equation}
where $Q(p|\lambda, \alpha)$ is given in (\ref{qufuiw}). After some algebraic manipulation, the estimator for $\hat{\alpha}_{LME}$ can be obtained as
\begin{equation}\label{lmomentse1}
\hat{\alpha}_{LME}=\frac{\log(2)}{\log (2)+\log\left(\sum_{i=1}^n (i-1)t_{(i)}\right)-\log\left(n(n-1)\bar{t}\right)} \cdot
\end{equation} 

Note that, substituting $\hat\alpha_{LME}$ in (\ref{lmoments1}) the estimator for $\hat{\lambda}_{LME}$ can be obtained by solving
\begin{equation}
\quad \quad \quad \hat\lambda_{LME}= \frac{\bar{t}^\alpha}{\Gamma^\alpha\left(1-\alpha^{-1}\right)}, \quad \quad \quad \alpha>1.
\end{equation} 

It is worth noting that, among the chosen methods, the  L-moments estimator was the only that has closed-form solution for both parameters.

\subsection{Ordinary and Weighted Least-Square Estimate}

 The least square estimators $\hat{\lambda}_{LSE}$ and $\hat{\alpha}_{LSE}$, can be obtained by minimizing
\begin{equation*}
S\left( \lambda, \alpha\right) = \sum_{i=1}^{n}\left[F\left( t_{(i)}\mid \theta ,\lambda
\right) - \frac {i}{n+1} \right]^{2},
\end{equation*}
with respect to $\lambda $ and $\alpha$. Similarly, they can also be obtained by solving the following non-linear equations (see Erto \cite{erto1989genesis} for more details):
\begin{eqnarray*}
&&\displaystyle\sum_{i=1}^{n}\left[ F\left( t_{(i)}\mid  \lambda, \alpha
\right) -\frac{i}{n+1}\right] \eta _{1}\left( t_{(i)}\mid  \lambda, \alpha
\right) =0, \\
&&\displaystyle\sum_{i=1}^{n}\left[ F\left( t_{(i)}\mid  \lambda, \alpha
\right) -\frac{i}{n+1}\right] \eta _{2}\left( t_{(i)}\mid  \lambda, \alpha\right) =0,
\end{eqnarray*}
where 
\begin{equation}  \label{delta1}
\eta_{1}\left( t_{(i)}\mid  \lambda, \alpha \right) =t_{(i)}^{-\alpha}\,e^{-\lambda t_{(i)}^{-\alpha}} \mbox{ and } \eta_{2}\left( t_{(i)}\mid  \lambda, \alpha \right) =\lambda t_{(i)}^{-\alpha}\log(t_{(i)})e^{-\lambda t_{(i)}^{-\alpha}} .
\end{equation}%

The weighted least-squares estimators (WLSEs), $\widehat{\lambda}_{WLSE}$ and $\widehat{\alpha}_{WLSE}$, can be obtained by minimizing
\begin{equation*}
W\left( \lambda, \alpha \right) = \sum_{i=1}^{n}
\frac {\left( n+1\right)^{2}\left( n+2\right)}{i\left( n-i+1\right)}
\left[ F\left( t_{(i)}\mid  \lambda, \alpha \right)- \frac {i}{n+1} \right]^{2}.
\end{equation*}
The estimators can also be obtained by solving the following non-linear equations,
\begin{eqnarray*}
&& \displaystyle \sum_{i=1}^{n}\frac {\left( n+1\right)^{2}\left( n+2\right)%
}{i\left( n-i+1\right)} \left[ F\left( t_{(i)}\mid  \lambda, \alpha \right) - 
\frac {i}{n+1} \right] \eta_{1}\left( t_{(i)}\mid  \lambda, \alpha \right) = 0,
\\
&& \displaystyle \sum_{i=1}^{n} \frac {\left( n+1\right)^{2}\left( n+2\right)%
}{i\left( n-i+1\right)} \left[ F\left( t_{(i)}\mid  \lambda, \alpha \right) - 
\frac {i}{n+1}\right] \eta_{2}\left( t_{(i)}\mid  \lambda, \alpha \right) = 0.
\end{eqnarray*}

\subsection{Method of Maximum Product of Spacings}

The maximum product of spacings (MPS) method is a powerful alternative to ML method for estimating the unknown parameters of continuous univariate distributions. This method was proposed by Cheng and Amin \cite{cheng1979maximum, cheng1983estimating}, and later independently developed by Ranneby \cite{ranneby1984maximum} as an approximation to the Kullback-Leibler measure of information.

 Let $D_{i}(\lambda, \alpha)=F\left( t_{(i)}\mid \lambda, \alpha \right)
-F\left( t_{(i-1)}\mid \lambda, \alpha \right)$, for  $i=1,2,\ldots
,n+1,$ be the uniform spacings of a random sample from the Fr\'echet distribution, where $F(t_{(0)}\mid \lambda, \alpha)=0$, $F( t_{(n+1)}\mid
\lambda, \alpha)=1$ and \
$\sum_{i=1}^{n+1} D_i (\lambda, \alpha) =1$. The maximum product of spacings estimators $\widehat{\lambda }_{MPS}$ and $ \widehat{\alpha }_{MPS}$ are obtained by maximizing the geometric mean of the spacings
\begin{equation}
G\left( \lambda, \alpha\right) =\left[ \prod\limits_{i=1}^{n+1}D_{i}( \lambda, \alpha)\right] ^{%
\frac{1}{n+1}},  \label{G}\nonumber
\end{equation}%
with respect to $\lambda$ and $\alpha$, or, equivalently, by maximizing  the logarithm of the geometric mode of sample spacings
\begin{equation}
H\left( \lambda, \alpha\right) =\frac{1}{n+1}\sum_{i=1}^{n+1}\log
D_{i} ( \lambda, \alpha).\nonumber
\end{equation}

The estimators  $\widehat{\lambda }_{MPS}$ and $ \widehat{\alpha }_{MPS}$
of the parameters $\lambda$ and $\alpha$ can  be obtained by solving
the following nonlinear equations%
\begin{equation}
\begin{aligned}
&\frac{\partial H\left( \lambda ,\alpha \right) }{\partial \lambda }  =\frac{1}{n+1}%
\sum\limits_{i=1}^{n+1}\frac{1}{D_{i}( \lambda,\alpha)} \left[ \eta_1
(t_{(i)} |  \lambda ,\alpha) - \eta_1 (t_{(i-1)} |  \lambda ,\alpha)
\right] =0, \\
&\frac{\partial H\left( \lambda ,\alpha \right)}{\partial \alpha }  =\frac{1}{%
n+1}\sum\limits_{i=1}^{n+1}\frac{1}{D_{i}( \lambda ,\alpha)} \left[
\eta_2 (t_{(i)} |  \lambda ,\alpha) - \eta_2 (t_{(i-1)} |  \lambda
,\alpha) \right] =0.\end{aligned}
\end{equation}

Note that if $t_{(i+k)}=t_{(i+k-1)}$ then $D_{i+k}(\lambda, \alpha)=D_{i+k-1}(\lambda, \alpha)=0$ for some $i$. Therefore, the MPS estimators are sensitive to closely spaced observations, especially ties. When ties are due to multiple observations, $D_{i}(\lambda, \alpha)$ should be replaced by the corresponding  likelihood $f(t_{(i)},\lambda, \alpha)$ since $t_{(i)}=t_{(i-1)}$.

Cheng and Amin \cite{cheng1983estimating} presented desirable properties of the MPS such as asymptotic efficiency and invariance. They also proved that the consistency of maximum product of spacing estimators holds under much more general conditions than for maximum likelihood estimators. Therefore, the MPS estimators are asymptotically normally distributed with a joint bivariate normal distribution given by
\begin{equation*}
(\hat{\lambda}_{MPS},\hat{\alpha}_{MPS})\sim N_2[(\lambda,\alpha),I^{-1}(\lambda,\alpha))],  \ \ \mbox{ for } n \to \infty .
\end{equation*}

\subsection{The Cram\'{e}r-von Mises maximum goodness-of-fit estimators}

The Cram\'{e}r-von Mises is a type of maximum goodness-of-fit estimators (also called minimum distance estimators) and is based on the difference between the estimate of the cumulative distribution function and the empirical distribution
function \cite{boos1981minimum}.

Macdonald \cite{macdonald1971estimation} motivated the choice of the Cram\'{e}r-von Mises statistic and provided empirical evidence that the bias of the estimator is smaller than the other goodness-of-fit estimators. The proposed estimator  is based on the Cram\'{e}r-von Mises statistics given by
\begin{equation*}
W_n^2=n\int_{-\infty}^{\infty}\left(F(t_{(i)})-E_n(t_{(i)}) \right)^2dF(t_{(i)}),
\end{equation*}  
where $E_n(\cdot)$ is the empirical density function. Boos \cite{boos1981minimum} discussed its asymptotic properties and  presented its computational form which is given by
\begin{equation}\label{cmvest}
C(\lambda ,\alpha )=\frac{1}{12n}+\sum_{i=1}^{n}\left( F\left( t_{(i)}\mid
\theta ,\lambda \right) -{\frac{2i-1}{2n}}\right) ^{2}.
\end{equation}%

The Cram\'{e}r-von Mises estimators $\widehat{\lambda }_{CME}$ and $\widehat{%
\alpha}_{CME}$ of the parameters $\lambda $ and $\alpha $ are obtained by minimizing 
\begin{eqnarray*}
\sum_{i=1}^{n}\left( F\left( t_{(i)}\mid \lambda ,\alpha \right) -{\frac{2i-1%
}{2n}}\right) \eta _{1}\left( t_{(i)}\mid \lambda ,\alpha  \right) &=&0, \\
\sum_{i=1}^{n}\left( F\left( t_{(i)}\mid \lambda ,\alpha  \right) -{\frac{2i-1%
}{2n}}\right) \eta _{2}\left( t_{(i)}\mid \lambda ,\alpha  \right) &=&0,
\end{eqnarray*}%
where $\eta _{1}\left( \cdot \mid \lambda ,\alpha  \right) $ and $\eta
_{2}\left( \cdot \mid \lambda ,\alpha  \right) $ are given 
respectively in (\ref{delta1}).

\subsection{Method of Anderson-Darling}
Other type of maximum goodness-of-fit estimators is based on an Anderson-Darling statistic and is known as the Anderson-Darling estimator. 
The Anderson-Darling statistic is given by
\begin{equation*}
\f{ADS}_n^2=n\int_{-\infty}^{\infty}\frac{\left(F(t_{(i)})-E_n(t_{(i)}) \right)^2}{F(t)(1-F(t))}dF(t_{(i)}).
\end{equation*}%
Boos \cite{boos1981minimum} also discussed the properties of the AD estimators and presented its computational form which the Anderson-Darling estimators $\widehat{\lambda }_{ADE}$ and $\widehat{\alpha}%
_{ADE}$  are obtained
by minimizing, with respect to $\lambda$ and $\alpha$, the function
\begin{equation}
A(\lambda ,\alpha) =-n-\frac{1}{n}\sum_{i=1}^{n}\left( 2i-1\right)
\left[\, -\lambda t_{(i)}^{-\alpha} + \log\left(1-e^{-\lambda t_{(n+1-i)}^{-\alpha}}\right)\right] .\nonumber
\end{equation}
These estimators can also be obtained by solving the following non-linear equations:
\begin{eqnarray*}
\sum_{i=1}^{n}\left( 2i-1\right) \left[ \frac{\eta _{1}\left( t_{(i)}\mid
\lambda ,\alpha \right) }{F\left( t_{(i)}\mid \lambda ,\alpha \right) }-%
\frac{\eta _{1}\left( t_{_{n+1-i:n}}\mid \lambda ,\alpha \right) }{1-
F\left( t_{n+1-i:n}\mid \lambda ,\alpha \right) }\right] &=&0, \\
\sum_{i=1}^{n}\left( 2i-1\right) \left[ \frac{\eta _{2}\left( t_{(i)}\mid
\lambda ,\alpha \right) }{F\left( t_{(i)}\mid \lambda ,\alpha \right) }-%
\frac{\eta _{2}\left( t_{_{n+1-i:n}}\mid \lambda ,\alpha \right) }{1-
F\left( t_{n+1-i:n}\mid \lambda ,\alpha \right) }\right] &=&0,
\end{eqnarray*}
where $\eta _{1}\left( \cdot \mid \lambda ,\alpha  \right) $ and $\eta
_{2}\left( \cdot \mid \lambda ,\alpha  \right) $ are given 
respectively in (\ref{delta1}).

\section{Bayesian Analysis}

In the previous sections, we have presented different estimation procedures using the frequentist approach. In this section, we consider the Bayesian inference approach for estimating the the unknown parameters of the Fr\'echet distribution. Bayesian analysis is an attractive framework in practical problems and has grown popularity in recent years. The prior distribution is a key part of the Bayesian inference and there are different types of priors distribution available in the literature (see, for instance, Ramos et al. \cite{ramos2017bayesian}). Usually, non-informative priors are preferable because if prior information on study parameters is unavailable or does not exist for a device, then initial
uncertainty about the parameters can be quantified with a non-informative prior distribution. An important non-informative reference prior was introduced by Bernardo \cite{bernardo1979a}, with further developments (see Bernardo \cite{bernardo2005} and references therein). The proposed idea was to maximize the expected Kullback-Leibler divergence between the posterior distribution and the prior. The reference prior provides posterior distribution with interesting properties, such as invariance under one-to-one transformations, consistent marginalization and consistent sampling properties. Recently, Abbas and Tang \cite{abbas2015analysis} derived two reference priors as well as the Jeffreys' prior for the Fr\'echet distribution. Here, we derive such priors by using different approach, the following proposition is used to obtain the reference priors.

\begin{proposition}\label{propositionop} [Bernardo \cite{bernardo2005}, pg 40, Theorem 14]. Let $\boldsymbol\theta=(\theta_1,\theta_2)$ be a vector of the ordered parameters of interest and $I(\theta_1,\theta_2)$ is the Fisher information matrix. If the parameter space  of $\theta_1$ does not depend of $\theta_2$ and $I_{j,j}(\boldsymbol\theta), j=1,2$ are factorized in the form 
\begin{equation*}
S_{1,1}^{-\frac{1}{2}}(\boldsymbol\theta)=f_1(\theta_1)g_1(\theta_2) \ \ \mbox{ and } \ I_{2,2}^{\frac{1}{2}}(\boldsymbol\theta)=f_2(\theta_1)g_2(\theta_2) . 
\end{equation*}
Then the reference prior for the ordered parameters $\boldsymbol\theta$ is given by $\pi_{\boldsymbol\theta}(\boldsymbol\theta)=f_1(\theta_1)g_2(\theta_2)$ and there is no need for compact approximations, even if the conditional priors are not proper.
\end{proposition} 

\begin{theorem} Let $\omega_1=(\lambda,\alpha)$  and $\omega_2=(\alpha,\lambda)$ be the vectors of the ordered parameters then the reference priors for the ordered parameters are respectively given by
\begin{equation}\label{priorref1}
\pi_{\omega_1}(\lambda,\alpha)\propto \frac{1}{\lambda\alpha} \quad \mbox{and} \quad
\pi_{\omega_2}(\alpha,\lambda)\propto \frac{1}{\lambda\alpha}.
\end{equation}
\end{theorem} 
\begin{proof} For $\omega_1=(\lambda,\alpha)$, we have $f_1(\lambda)=\lambda^{-1}$, $g_1(\alpha)=1$, $\,f_2(\lambda)=\sqrt{\frac{\pi^2}{6} +\left(1-\gamma-\log(\lambda)\right)^2},$ $g_2(\alpha)=\alpha^{-1}$. Therefore, $\pi_{\omega_1}(\lambda,\alpha)\propto f_1(\lambda)g_2(\alpha)\propto\lambda^{-1}\alpha^{-1}$. Analogously, we obtain $\pi_{\omega_2}(\lambda,\alpha)$.
\end{proof} 

Moreover, Berger et al. \cite{berger2015} suggested by starting with a collection of reference priors and then taking the arithmetic mean or the geometric mean to obtain an overall reference prior. Therefore, an overall reference prior is the same as (\ref{priorref1}).
Another well-known non-informative prior was introduced by Jeffreys \cite{jeffreys1946invariant} and can be obtained through the square root of the determinant of Fisher information matrix (\ref{mfishergg}), such prior is widely used due to its invariance property under one-to-one transformations. After some algebraic manipulations we have 
\begin{equation}\label{priorjefrreys}
\pi_J(\lambda,\alpha)\propto \frac{1}{\lambda\alpha},
\end{equation}
which is equal to (\ref{priorref1}). It is worth noting that such prior also arises considering the Jeffreys rule (see, Kass and Wasserman \cite{kass1996selection}). Therefore, even considering different methods to obtain non-informative priors, we have the same prior for the Fr\'echet distribution.

Combining the likelihood (\ref{loglikeiw}) and the prior (\ref{priorjefrreys}), the posterior distribution is
\begin{equation}\label{postphi1}
\pi(\lambda,\alpha|\boldsymbol{t})=\frac{\left(\lambda\alpha\right)^{n-1}}{c(\boldsymbol{t})}
\prod_{i=1}^n t_i^{-\alpha}\exp\left\{-\lambda\sum_{i=1}^n t_i^{-\alpha}\right\},
\end{equation}
where
\begin{equation}\label{posteriord2}
c(\boldsymbol{t})=\int\limits_{\mathcal{A}}\left(\lambda\alpha\right)^{n-1}
\prod_{i=1}^n t_i^{-\alpha}\exp\left\{-\lambda\sum_{i=1}^n t_i^{-\alpha}\right\}d\boldsymbol{\theta}
\end{equation}
and $\mathcal{A}=\{(0,\infty)\times(0,\infty)\}$ is the parameter space of $\boldsymbol{\theta}$. To get reliable inference, first we have to check whether the posterior distribution (\ref{postphi1}) is a proper posterior, i.e., $c(\boldsymbol{t})<\infty$. 

\begin{theorem} The posterior distribution (\ref{postphi1}) is proper if $n\geq2$.
\end{theorem} 
\begin{proof} Since $\left(\lambda\alpha\right)^{n-1}
\prod_{i=1}^n t_i^{-\alpha}\exp\left\{-\lambda\sum_{i=1}^n t_i^{-\alpha}\right\}\geq 0$, by Tonelli theorem (see Folland \cite{folland}) we have
\begin{equation*}
\begin{aligned}
c(\boldsymbol{t})&=\int\limits_{\mathcal{A}}\left(\lambda\alpha\right)^{n-1}
\prod_{i=1}^n t_i^{-\alpha}\exp\left\{-\lambda\sum_{i=1}^n t_i^{-\alpha}\right\}d\boldsymbol{\theta} \\
&= \int\limits_0^{\infty}\int\limits_0^{\infty}\left(\lambda\alpha\right)^{n-1}
\prod_{i=1}^n t_i^{-\alpha}\exp\left\{-\lambda\sum_{i=1}^n t_i^{-\alpha}\right\}d\lambda d\alpha \\
&=\Gamma(n)\int\limits_0^{\infty} \alpha^{n-1}
\prod_{i=1}^n t_i^{-\alpha}\left\{\sum_{i=1}^n t_i^{-\alpha}\right\}^{-n} d\alpha \\ & \leq \int\limits_0^{\infty} \alpha^{n-1}
\prod_{i=1}^n \left(\frac{t_i}{\min({t_1,\ldots,t_n})}\right)^{-\alpha} d\alpha < \infty,
\end{aligned}
\end{equation*}
the last inequality holds if $n\geq 2$.
\end{proof}
 
The marginal posterior distribution for $\alpha$ is given by
\begin{equation}\label{postc1}
\pi(\alpha|\,\boldsymbol{t})\propto \alpha^{n-2}
\prod_{i=1}^n t_i^{-\alpha}\left\{\sum_{i=1}^n t_i^{-\alpha}\right\}^{-n}.
\end{equation}

The conditional  posterior distribution for $\lambda$ is 
\begin{equation}\label{postc2}
\pi(\lambda|\alpha,\boldsymbol{t})\sim \f{Gamma}\left(n,\sum_{i=1}^n t_i^{-\alpha}\right),
\end{equation}
where Gamma$(a, c)$ is the Gamma distribution with PDF, $f(x,a,c)=c^a\,x^{a-1}\exp(-c
x)/\Gamma(a)$. By considering  marginal and conditional  posterior distributions (\ref{postc1}) and (\ref{postc2}), the convergence of Markov Chain Monte Carlo (MCMC) method can be easily achieved. 

Abbas and Tang \cite{abbas2015analysis} derived the same priors (\ref{priorref1}) and proved that the obtained posterior is proper. However, the authors did not proved that the obtained posterior means for $\alpha$ and $\lambda$ are finite. These proofs are important in order to obtain reliable results. The following theorem prove that the posterior mean of $\lambda$ is improper depending on the data.

\begin{theorem}\label{theorempsmean} The posterior mean for $\lambda$ is improper in case $\prod_{i=1}^n \left(\frac{t_i}{\min({t_1,\ldots,t_n})}\right)\leq \min({t_1,\ldots,t_n})$ and $n\geq 2$.
\end{theorem} 
\begin{proof} The posterior mean for $\lambda$ is given by
\begin{equation*}\hat{\lambda}(\boldsymbol{t}) =\int\limits_{\mathcal{A}}\lambda\left(\lambda\alpha\right)^{n-1}
\prod_{i=1}^n t_i^{-\alpha}\exp\left\{-\lambda\sum_{i=1}^n t_i^{-\alpha}\right\}d\boldsymbol{\theta} .
\end{equation*}
Notice that, since ${\prod_{i=1}^n t_i }/{\min({t_1,\ldots,t_n})^{n+1}}\leq 1$ by hypothesis it follows that $\left({\prod_{i=1}^n t_i }/{\min({t_1,\ldots,t_n})^{n+1}}\right)^{-\alpha}\geq 1^{-\alpha} = 1$.

Now, since $\lambda\left(\lambda\alpha\right)^{n-1} \prod_{i=1}^n t_i^{-\alpha}\exp\left\{-\lambda\sum_{i=1}^n t_i^{-\alpha}\right\}\geq 0$, by Tonelli theorem we have
\begin{equation*}
\begin{aligned}
\hat{\lambda}(\boldsymbol{t})&=\int\limits_{\mathcal{A}}\lambda\left(\lambda\alpha\right)^{n-1}
\prod_{i=1}^n t_i^{-\alpha}\exp\left\{-\lambda\sum_{i=1}^n t_i^{-\alpha}\right\}d\boldsymbol{\theta} \\
&= \int\limits_0^{\infty}
\prod_{i=1}^n t_i^{-\alpha}\int\limits_0^{\infty}\alpha^{n-1} \lambda^n\exp\left\{-\lambda\sum_{i=1}^n t_i^{-\alpha}\right\}d\lambda d\alpha \\
&=\Gamma(n+1) \int\limits_0^{\infty}  \alpha^{n-1}
\prod_{i=1}^n t_i^{-\alpha}\left\{\sum_{i=1}^n t_i^{-\alpha}\right\}^{-n} d\alpha \\ &
\geq \Gamma(n+1)n^{-n}\int\limits_0^{\infty} \alpha^{n-1} \left(\frac{\prod_{i=1}^n t_i }{\min({t_1,\ldots,t_n})^{n+1}}\right)^{-\alpha} d\alpha 
\\ 
&\geq \Gamma(n+1)n^{-n} \int\limits_0^{\infty} \alpha^{n-1} d\alpha =  \infty,
\end{aligned}
\end{equation*}
where the last inequality holds for $n\geq 2$.

\end{proof}

From the above theorem, we can see that the posterior mean of $\lambda$ may be improper depending on the data, which is undesirable. The use of Monte Carlo methods in improper posterior was discussed by Hobert and  Casella \cite{hobert1996effect}. The authors argued that ``one  can not expect the Gibbs output to provide a ``red flag", informing the user that the posterior is improper. The user must demonstrate propriety before a Markov Chain Monte Carlo technique is used.'' In our case, as the properness of the posterior mean of $\lambda$ depend on the data, other alternative is needed to be considered as Bayes estimator. A discussion about other alternatives will be presented in the next section.

\section{Simulation Study}

In this section, we present some experimental results to evaluate the performance of the different  methods
of estimation discussed in the previous sections.

\subsection{Classical approach}\label{anacla}

In this subsection, we compare the efficiency of the different estimation methods considering  classical approach. The following procedures are adopted:
\begin{enumerate}
\item Generate N samples from the Fr\'echet$(\lambda,\alpha)$ distribution with size $n$ and compute $\boldsymbol{\hat\theta}=(\hat{\lambda},\hat{\alpha})$ via MLE, ME, LME ,LSE, WLSE, PCE, MPS, CME and ADE.
\item Using $\boldsymbol{\hat\theta}$ and $\boldsymbol{\theta}$, compute the mean relative estimates (MRE)= $\sum_{i=1}^{N}({\hat\theta_i/\theta_j})/{N}$ and the mean squared errors (MSE)= $\sum_{i=1}^{N}{(\hat\theta_i-\theta_j)^2}/{N}$, for $j=1,2$.
\end{enumerate}

Considering the above approach, the most efficient estimator will be the one whose MREs is closer to one with smaller MSEs. These results are computed using the software R (see, R Core Team \cite{rteam}). The seed used to generate the pseudo-random samples from the Fr\'echet distribution was 2018. We have chosen the values to perform this procedure are $N=500,000$ and $n=(15,20,25,\cdots,140)$. We presented results only for $\boldsymbol\theta=(2,4)$ due to space constraint. However, results are similar for other choices of $\lambda$ and $\alpha$. Figure \ref{fsimulation1} shows the MREs, MSEs for the estimates of $\boldsymbol\theta$.  The horizontal lines in the Figure correspond to MREs and MSEs being one and zero respectively.  
	\begin{figure}[!h]
	\centering
	\includegraphics[scale=0.75]{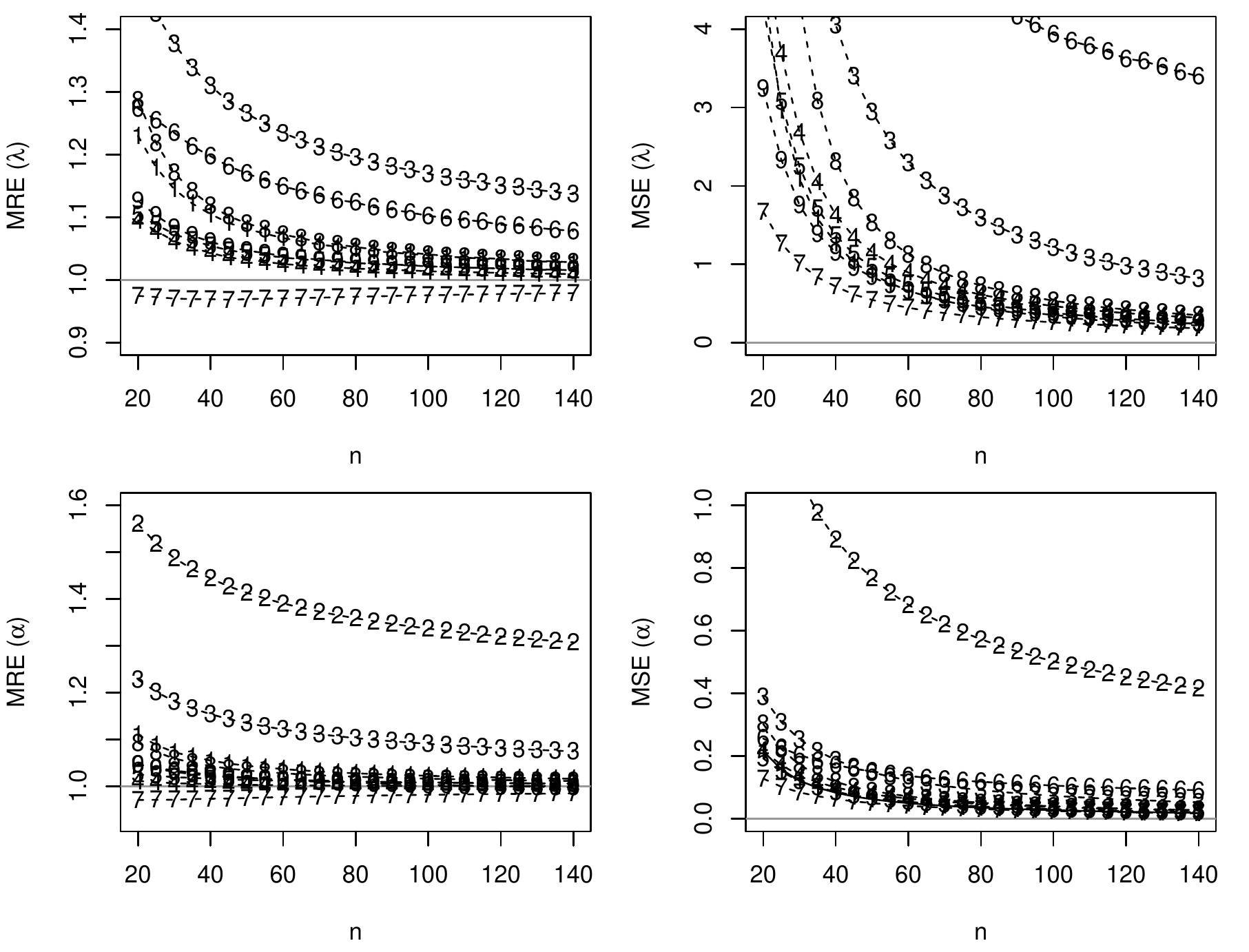}
	\caption{MREs, MSEs for the estimates of $\lambda=2$ and $\alpha=4$ for $N=500,000$ simulated samples, considering different values of $n$  using the following estimation method 1-MLE, 2-ME,3-LME ,4-LSE, 5-WLSE, 6-PCE, 7-MPS, 8-CME, 9-ADE.}\label{fsimulation1}
	\end{figure}

From  Figure 1, we observe that the MSEs of all estimators of the parameters tend to zero for large $n$ and also, the values of
MREs tend to one, i.e. the estimators are asymptotically unbiased and consistent for the parameters. For both parameters, we observe that the moment estimators has the largest MREs and MSEs respectively among all the considered estimators. Further, we also observe that the MPS, ADE, LSE and WLSE performs better than the MLEs for small and moderate sample sizes in terms of MREs and MSEs. Moreover, the MPS estimators have the smallest MSEs among all the considered estimators.

Combining all results with the good properties of the MPS method such as consistency, asymptotic efficiency, normality and invariance, we suggest to use MPS estimators of the parameters of Fr\'echet distribution in all practical purposes.

\subsection{Bayesian approach}

In this subsection, we obtain the Bayes estimator under the same assumptions of section \ref{anacla}. The $95\%$ coverage probability of the asymptotic confidence intervals under the classical set-up and the credible intervals (CI$_{95\%}$) under the Bayesian set-up are also evaluated. For large number of experiments and considering  confidence level of $95\%$, the frequencies of intervals that covered the true values of $\boldsymbol{\theta}$ should be closer to $95\%$. 

\subsubsection{Metropolis-Hastings (M-H) algorithm}

Since the marginal posterior distribution of $\alpha$ does not have closed form, the Metropolis-Hastings (M-H) algorithm is applied to generate samples from this marginal distribution. In this case, we have used the Gamma distribution as transition kernel $q\left(\alpha^{(j)}|\alpha^{(*)},b\right)$ for sampling values of $\alpha$, where b is a known hyperparameter that controls the acceptance rate of the algorithm. It is worth mentioning that, the choice of the transition kernel is arbitrary, and other non-negative random variables could also be used instead. The M-H  algorithm operates as follows:

\begin{enumerate}

\item Start with an initial value $\alpha^{(1)}$ and set the iteration counter $j=1 ;$

\item Generate a random value $\alpha^{(*)}$ from the proposal $\f{Gamma}(\alpha^{(j)},b)$;

\item Evaluate the acceptance probability
\begin{equation*}
h\left(\alpha^{(j)},\alpha^{(*)}\right)=\min\left(1, \frac{\pi\left(\alpha^{(*)}|\boldsymbol{t}\right)}{\pi\left(\alpha^{(j)}|\boldsymbol{t}\right)} \frac{\f{q}\left(\alpha^{(j)},\alpha^{(*)},b\right)}{\f{q}\left(\alpha^{(*)},\alpha^{(j)},b\right)}\right),
\end{equation*}
where $\pi(\cdot)$ is given in  (\ref{postc1}). Generate a random value $u$ from an independent uniform in $(0,1)$;

\item If $h\left(\alpha^{(j)},\alpha^{(*)}\right)\geq u(0,1)$ then $\alpha^{(j+1)}=\alpha^{(*)}$, otherwise $\alpha^{(j+1)}=\alpha^{(j)}$;

\item Change the counter from $j$ to $j + 1$ and return to step 2 until convergence
is reached.
\end{enumerate}

In this case, we choose b to be equal to one. However, other values can also be considered. To decrease the necessary time taken for M-H method to reach the convergence, we can use (\ref{lmomentse1}) as a good initial value for $\alpha^{(1)}$. Considering the conditional posterior distribution (\ref{postc2}), the Bayes estimator for $\lambda$ can be obtained direcly from the Gamma distribution with $\left(n,\sum_{i=1}^n t_i^{-\hat{\alpha}_{Bayes}}\right)$ as well as its respective credible interval is evaluated by the quantile function. The decision rule used to obtain the Bayes estimators will be presented in the following.
	
\subsubsection{Bayes estimator}

The selection of a decision rule to obtain posterior estimates is of fundamental problem in Bayesian statistics. Usually, this problem is over looked by many authors. The most common risk function used to obtain the Bayes estimates is the mean squared error, by considering this risk function, the obtained Bayes estimates are the posterior means. Other alternative functions can be considered, for instance, the posterior mode, also known as maximum a posteriori probability (MAP) is obtained by assuming a 0-1 loss function, while the posterior median is obtained considering a linear loss function. 

For the Fr\'echet distribution, Abbas and Tang \cite{abbas2015analysis} presented two Bayes estimators, the first is the MAP estimators obtained from Laplace's approximation. Despite of the fact that reference priors are invariant under-one-to-one transformation, the MAP is not \cite{murphy2012machine}, and therefore, the obtained Bayes estimator is not invariant. Although the authors  considered other Bayes estimator, but they did not specify how the proposed estimator was obtained. This lack of information makes such results non-reproducible, which is undesirable for any application. As the authors used Monte Carlo methods, the most common Bayesian estimator is the posterior mean. However, from Theorem \ref{theorempsmean} we have proved that the posterior mean may be improper depending on the data, which is undesirable. In fact, consider the example analyzed by Abbas and Tang \cite{abbas2015analysis} related to fatigue lifetime data. The data is given by: 152.7, 172.0, 172.5, 173.3, 193.0, 204.7, 216.5, 234.9, 262.6, 422.6. From the proposed theorem, the posterior mean of $\lambda$ is improper if
\begin{equation*}
\prod_{i=1}^n \left(\frac{t_i}{\min({t_1,\ldots,t_n})}\right)\leq \min({t_1,\ldots,t_n}),
\end{equation*}
since $25.4\leq 152.7$, then the posterior mean of $\lambda$ is improper and can not be used. As it was discussed earlier, for the proposed dataset we can easily draw a sample for the marginal distribution and compute the posterior mean without any ``red flag'', but such estimate is meaningless. Therefore, for the Fr\'echet distribution the posterior median is a reasonable choice as it is invariant under-one-to-one transformation and finite for $n\geq2$ almost surely.

\subsubsection{Results}	

For each simulated data set under the Bayesian approach, $5,500$ iterations are performed using the MCMC methods. As a burning sample, we discarded the first $500$ initial values taking jumps of size $5$  to reduce the auto-correlation values among the chain, getting at the end one chain of size $1,000$. To validate the convergence of the obtained chain, we used the Geweke criterion \cite{geweke1991evaluating} with a $95\%$ confidence level. At the end, $10,000$ posterior medians for $\alpha$ and $\lambda$ were computed.
	\begin{figure}[!htb]
	\centering
	\includegraphics[scale=0.75]{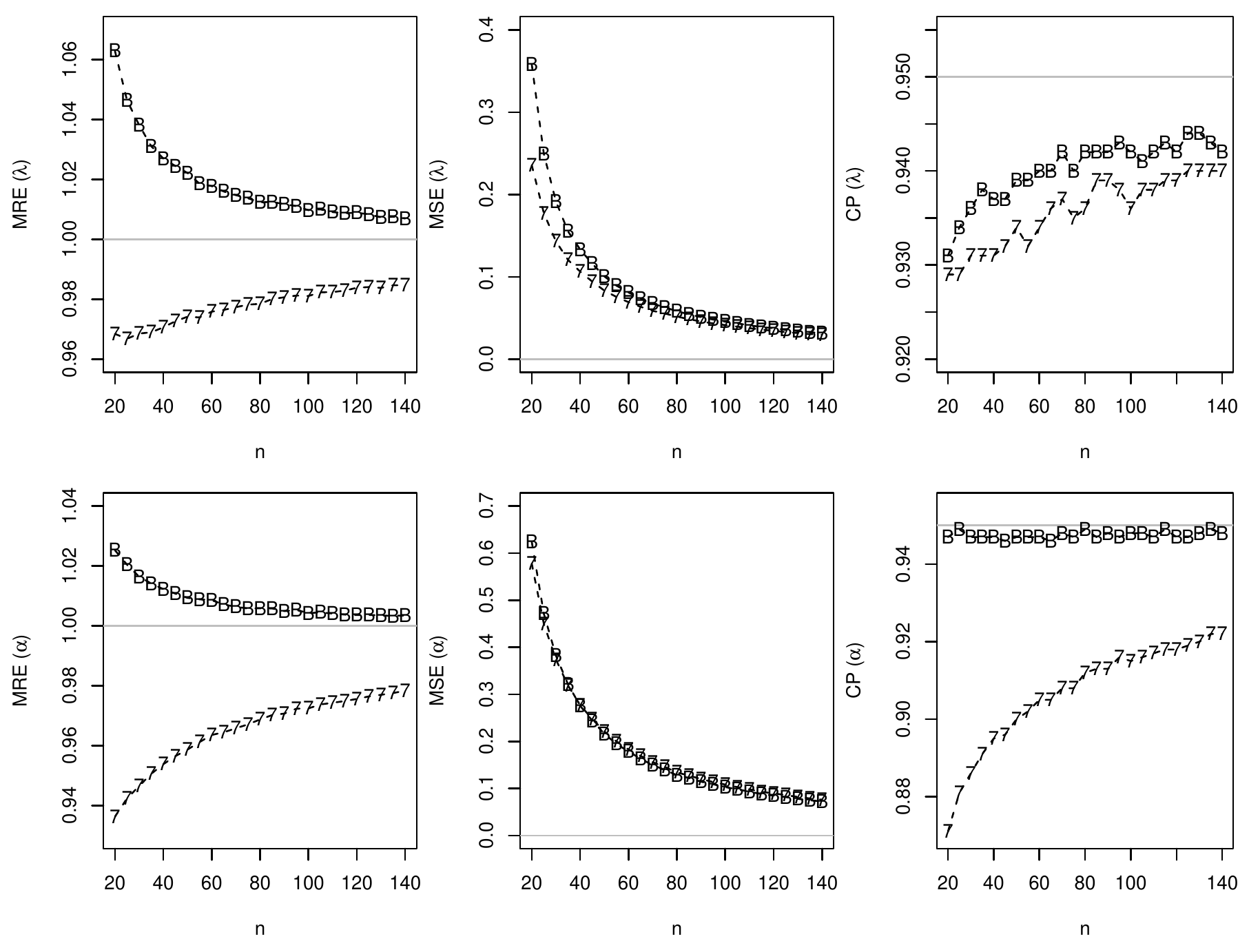}
	\caption{MREs, MSEs and the coverage probability with a $95\%$ confidence level for the estimates of $\lambda=2$ and $\alpha=4$ for $N=500,000$ simulated samples, considering different values of $n$  using the following estimation method, 7 - MPS, B - Bayes estimator.}\label{fsimulation2}
	\end{figure}
	
 The results are presented only for $\boldsymbol\theta=(2,4)$ due to space constraint. However, the following results are also similar for other choices of $\lambda$ and $\alpha$. Figure \ref{fsimulation2} presents the MREs, MSEs and the coverage probability with a confidence level equal to $95\%$ for the estimates obtained for the MPS and posterior medians. It is to be noted that, since MPS performs better than their counterparts,  we compare Bayes estimates with MPS.

From Figure \ref{fsimulation2} we observe that the Bayes estimates have the MRE nearest to one while both approaches ( MPS and Bayes estimates)  have approximately the same MSE. Moreover, considering the Bayesian approach, we have obtained accurate coverage probability through the  highest posterior density intervals. Therefore, we conclude that the Bayesian approach is preferred in order to make inference on  unknown parameters of the Fr\'echet distribution.

\section{Applications}\label{aplicwe}

In this section, we analyze five real data sets related to  minimum monthly flows of water (m$^3$/s) on the Piracicaba River, located in S\~ao Paulo state, Brazil. This study can be useful to protect and maintain aquatic resources for the state \cite{tennant1976instream, reiser1989status}. The data sets (see in Appendix B for more details) are obtained from the Department of Water Resources and Power agency manager of water resources of the State of S\~ao Paulo from 1960 to 2014.

Table \ref{tableest} presents the posterior median (Bayes estimators) and $95\%$ credible intervals for $\lambda$ and $\alpha $ of the Fr\'echet distribution for the  data sets related to the total monthly rainfall during May, June, July, August and September at Piracicaba River.

\begin{table}[!h]
\caption{Bayes estimates, $95\%$ credible intervals for $\lambda$ and $\alpha$ for the data sets related to the total monthly rainfall during May, June, July, August and September at Piracicaba River.}\centering 
\par
\begin{center}
\begin{tabular}{c|c|r|c}
\hline
Month & $\boldsymbol{\theta}$ & $\boldsymbol{\hat{\theta}}_{Bayes}$ & CI$_{95\%}(\boldsymbol{\theta})$ \\ \hline
\multirow{2}{*}{May} & \ \ $\lambda$ \ \ & 309.890 & (223.248; 416.505)  \\ 
& \ \ $\alpha$ \ \   &  1.817 & (1.401;  2.293) \\ \hline
\multirow{2}{*}{June} & \ \ $\lambda$ \ \  & 89.758 & (64.376; 121.075) \\ 
& \ \ $\alpha$ \ \  & 1.585  & (1.194; 2.033) \\ \hline
\multirow{2}{*}{July} & \ \ $\lambda$ \ \  &  204.493 & (146.666; 275.840) \\ 
& \ \ $\alpha$ \ \  &  2.048 & (1.549; 2.637) \\ \hline
\multirow{2}{*}{August} & \ \ $\lambda$ \ \  & 401.656 & (290.594; 537.969) \\ 
& \ \ $\alpha$ \ \  & 2.4585 & (1.876; 3.119) \\ \hline
\multirow{2}{*}{September} & \ \ $\lambda$ \ \  & 55.128 & (39.539; 74.362) \\ 
& \ \ $\alpha$ \ \  & 1.529 & (1.163; 1.939) \\ \hline
\end{tabular}\label{tableest}
\end{center}
\end{table}

The results obtained using the Fr\'echet distribution are compared with the Weibull, Gamma, Lognormal (LN), Gumbel and Generalized Exponential (GE) distribution \cite{gupta2001generalized}. We consider certain discrimination criteria  such as  BIC (Bayesian Information Criteria), AIC (Akaike Information Criteria) and AICc (Corrected Akaike information criterion) and computed respectively by BIC$=-2l(\hat{\boldsymbol{\theta}};\boldsymbol{t})+k\log(n)$, AIC$=-2l(\hat{\boldsymbol{\theta}};\boldsymbol{t})+2k$ and AICc$=$AIC$+[{2\,k\,(k+1)}]/[{(n-k-1)}]$, where $k$ is the number of parameters to be fitted and $\hat{\boldsymbol{\theta}}$ is estimation of $\boldsymbol{\theta}$. Given a set of candidate models for $\boldsymbol{t}$, the preferred model is the one which provides the minimum values of the aforementioned statistics.

Table \ref{tabelaic} presents the results of BIC, AIC and  AICc, for different probability distributions. The goodness of fit can also be checked through the over plot of the survival function adjusted by the proposed theoretical models onto the empirical survival function as shown  in Figure \ref{grafico-obscajust1}.
\begin{table}[!h]
\caption{Results of the BIC,  AIC and the AICc for different probability distributions for the data sets related to the minimum flows of water (m$^3$/s) during May, June, July and August at Piracicaba River.}\centering 
\par
\begin{center}
\begin{tabular}{c|c|r|r|r|r|r|r}
\hline
\multicolumn{1}{c|}{Month}  & \multicolumn{1}{c|}{Test}  & \multicolumn{1}{c|}{Fr\'echet} & \multicolumn{1}{c|}{Weibull} & \multicolumn{1}{c|}{Gamma} & \multicolumn{1}{c|}{LN} & \multicolumn{1}{c|}{Gumbel} & \multicolumn{1}{c}{GE} \\ \hline
\multirow{3}{*}{May} & \ \ BIC \ \  & \textbf{361.43} & 390.91 & 386.90 & 369.93 & 394.23 & 384.04 \\ 
& \ \ AIC \ \                       & \textbf{358.06} & 387.53 & 383.52 & 366.55 & 390.85 & 380.66  \\ 
& \ \ AICc \ \                      & \textbf{358.38} & 387.86 &  383.84 &  366.88 & 391.18 &  380.98 \\  \hline
\multirow{3}{*}{June} & \ \ BIC \ \ & \textbf{346.92} & 379.72 & 381.02 & 359.70 & 403.55 & 380.81 \\ 
& \ \ AIC \ \                       & \textbf{343.60} & 376.39 & 377.69 & 356.37 & 400.22 &  377.48  \\ 
& \ \ AICc \ \                      & \textbf{343.93} &  376.73 &  378.03 &  356.70 & 400.55 &  377.81 \\  \hline
\multirow{3}{*}{July} & \ \ BIC \ \ & \textbf{302.86} & 336.50 & 332.78 & 316.75 & 341.31 & 330.30 \\
& \ \ AIC \ \                       & \textbf{299.54} &  333.17 & 329.45 & 313.42 & 337.98 &  326.97  \\ 
& \ \ AICc \ \                      & \textbf{299.87} & 333.50 &  329.79 &  313.75 & 338.32 &  327.30 \\  \hline
\multirow{3}{*}{August} & \ \ BIC \ \ & \textbf{283.92} & 310.33 & 303.41 & 294.30 & 303.68 & 299.35 \\ 
& \ \ AIC \ \                         & \textbf{280.49} & 306.90 & 299.98 & 290.87 & 300.25 &  295.92  \\ 
& \ \ AICc \ \                        & \textbf{280.81} & 307.22 &  300.30 &  291.19 &  300.57 &  296.24 \\  \hline
\multirow{3}{*}{September} & \ \ BIC \ \ & \textbf{329.06} & 344.21 & 341.77 & 332.96 & 351.45 & 340.68 \\
& \ \ AIC \ \                            & \textbf{325.73} & 340.89 & 338.44 & 329.63 & 348.12 &  337.35  \\ 
& \ \ AICc \ \                           & \textbf{326.06} & 341.22 &  338.77 & 329.96 & 348.45 & 337.69 \\  \hline
\end{tabular}\label{tabelaic}
\end{center}
\end{table}

\begin{figure}[!h]
\centering
\includegraphics[scale=0.68]{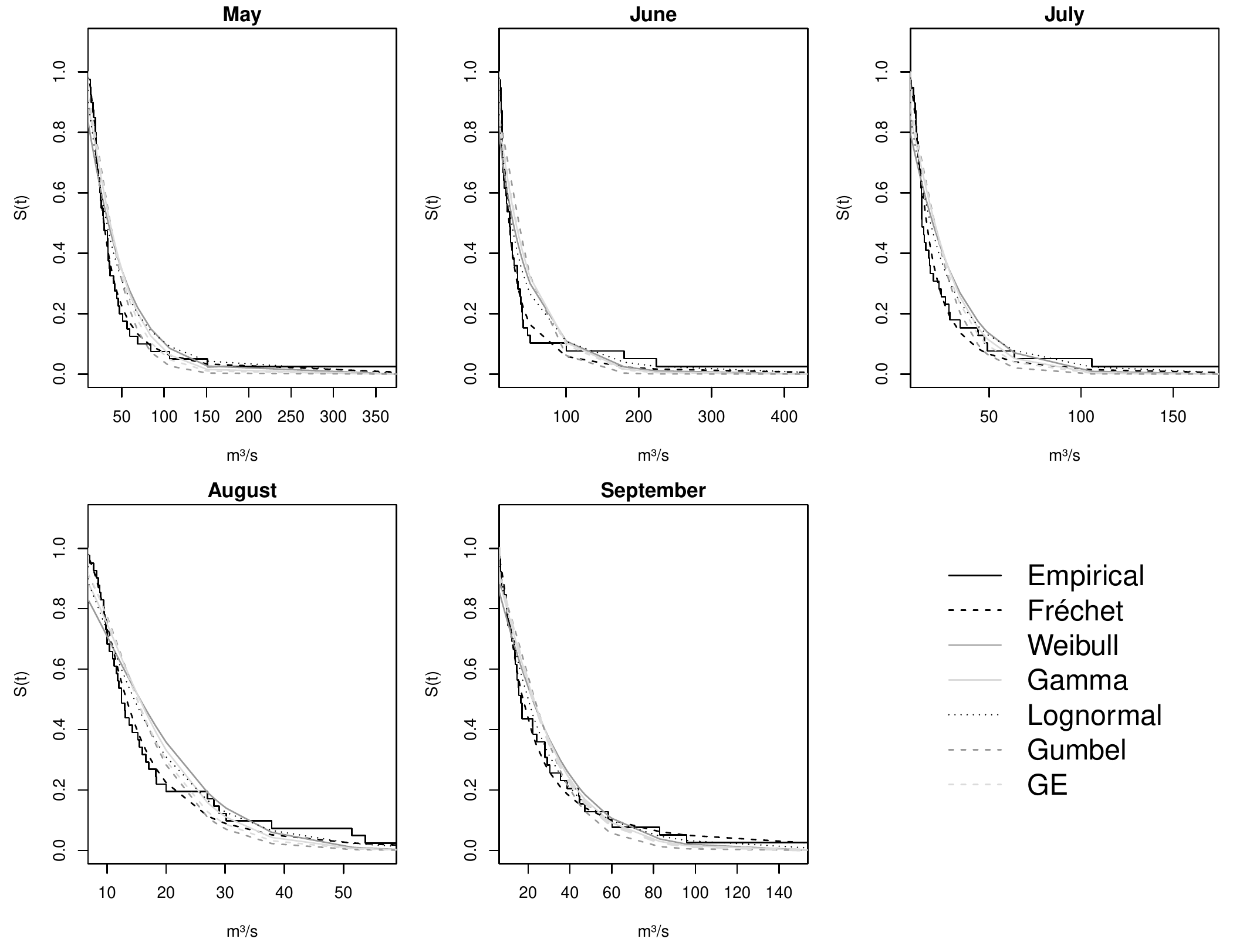}
\caption{Survival function adjusted by the empirical and the different distributions for the data sets related to the minimum flows of water (m$^3$/s) during May till September at Piracicaba River.}
\label{grafico-obscajust1}
\end{figure}

Comparing the empirical survival function with the adjusted distributions, we observe that Fr\'echet distribution fits better than the chosen models. These results are confirmed from the BIC, AIC and AICc values, since Fr\'echet distribution has the minimum values for the proposed data sets. Therefore,  our proposed methodology can be used successfully to analyze the minimum flow of water during May, June, July,  August and September at Piracicaba River  using the Fr\'echet distribution with the Bayesian approach.

\section{Conclusions}

In this paper, we considered different methods of estimation of the unknown parameters both from frequentist and Bayesian viewpoint of 
Fr\'echet distribution. We considered the MLEs, MMEs, LME, PCEs, LSEs, WLSE, MPS, CME and ADE as frequentist estimators. As it is
not feasible to compare these methods of estimation theoretically, we have performed an extensive simulation study to compare these nine methods of estimation. We compared these frequentist estimators mainly with respect to MREs and MSEs. The results show that among the classical estimators, MPS performs better than their counterpart in terms of  MSEs in the simulation study. Additionally, we have obtained Bayes estimates through posterior median and compared with respect to MREs and MSEs. We have obtained accurate coverage probability through the credible intervals. Therefore, combining all these results, we conclude that the Bayesian approach is preferred in order to make inference on unknown parameters of the Fr\'echet distribution. We also compared this model with some existing competing models and the Fr\'echet distribution performed reasonably well.

\section*{Acknowledgements}

The authors are thankful to the Editorial Board and to the reviewers for their valuable comments and suggestions which led to this improved version. 
The research was partially supported by CNPq, FAPESP and CAPES of Brazil.

\section*{Disclosure statement}

No potential conflict of interest was reported by the author(s)

\bibliographystyle{tfs}

\bibliography{interacttfssample}

\appendix

\section{Proof of Theorem \ref{themleune}.}

\begin{proof} It is important to notice that the solutions to the likelihood equation (\ref{likelihood1}) and (\ref{likelihood2}) may fail to be the global maximum of $l(\lambda,\alpha | t)$. However, the reasoning that follows justify why in this case these solution will in fact be the global maximum. Notice that
\begin{equation*}
\max_{ \lambda\in \mathbb{R}^+,\ \alpha\in \mathbb{R}^+} l(\lambda, \alpha|t)= \max_{\alpha\in \mathbb{R}^+ }\left(\max_{\lambda\in \mathbb{R}^+} l(\lambda, \alpha|t)\right)
\end{equation*}

To find $\max_{\lambda\in \mathbb{R}^+} l(\lambda, \alpha|t)$ we notice its derivative in relation to $\lambda$, for $\alpha$ fixed, is given by $\frac{\partial l(\lambda,\alpha | t)}{\partial \lambda} = \frac{n}{\lambda} -\sum_{i=1}^n t_i^{-\alpha}$ which is positive for $\lambda < \widehat{\lambda }(\alpha) = \frac{n}{\sum_{i=1}^n t_i^{-\alpha}}$ and negative in case $\lambda > \widehat{\lambda }(\alpha)$. Therefore $\widehat{\lambda}(\alpha)$ is the unique value that provides $\max_{\lambda\in \mathbb{R}^+} l(\lambda, \alpha|t)$ for $\alpha$ fixed and
\begin{equation*}
\max_{\lambda\in \mathbb{R}^+} l(\lambda, \alpha|t) = l\left(\widehat{\lambda }(\alpha),\alpha\right)
\end{equation*}

Now, to find $\max l\left(\widehat{\lambda}(\alpha),\alpha | t\right)$ notice, by the chain rule, that $l\left(\hat{\lambda}(\alpha),\alpha | t\right)$ is a differentiable function in relation to $\alpha$ with derivative given by $\frac{d l(\lambda(\alpha),\alpha | t)}{d \alpha} =  \frac{\partial l(\lambda(\alpha),\alpha|t)}{\partial \lambda}\lambda'(\alpha) + \frac{\partial l(\lambda(\alpha),\alpha | t)}{\partial \alpha} = \frac{\partial l(\lambda(\alpha),\alpha | t)}{\partial \alpha}$ because $\frac{\partial l(\lambda(\alpha),\alpha | t)}{\partial \lambda} = 0$. Let $G(\alpha) = \frac{d l(\lambda(\alpha),\alpha | t)}{d \alpha}$

Let us prove that
\begin{equation*}
G(\alpha )=\frac{n}{\alpha }-\sum_{i=1}^{n}\log t_{i}+\frac{n}{
\sum\limits_{i=1}^{n}t_{i}^{-\alpha }}\sum_{i=1}^{n}t_{i}^{-\alpha }\log t_{i}
\end{equation*}%
admits a unique solution. It is straightforward to see that $\lim {}_{\alpha
\rightarrow 0+}G(\alpha )=\infty$. We now prove that $\lim {}_{\alpha
\rightarrow \infty}G(\alpha )$ exists and is negative. Let $u_i = \frac{t_i}{\sqrt[n]{\prod_{i=1}^n t_i}}$ for $i=1\cdots, n$. Then $t_i=u_i \sqrt[n]{\prod_{i=1}^n t_i}$ for $i=1,\cdots, n$ and therefore
\begin{equation*}
\begin{aligned}
G(\alpha )&=\frac{n}{\alpha }-\sum_{i=1}^{n}\log t_{i}+\frac{n}{
\sum\limits_{i=1}^{n}u_{i}^{-\alpha }}\sum_{i=1}^{n}u_{i}^{-\alpha }\log\left(u_i \sqrt[n]{\prod_{i=1}^n t_i}\right) \\
&=\frac{n}{\alpha } -\sum_{i=1}^{n}\log t_{i} + \log{{\prod_{i=1}^n t_i}} + \frac{n}{
\sum\limits_{i=1}^{n}u_{i}^{-\alpha }}\sum_{i=1}^{n}u_{i}^{-\alpha }\log\left(u_i\right) \\
&=\frac{n}{\alpha } + \frac{n}{
\sum\limits_{i=1}^{n}u_{i}^{-\alpha }}\sum_{i=1}^{n}u_{i}^{-\alpha }\log\left(u_i\right).
\end{aligned}
\end{equation*}%
Now let $u_{\min} = \min{u_1,\cdots,u_n}$, and $r$ be the number of times $u_{\min}$ appears in $u_1,\cdots,u_m$. Then $\lim_{\alpha\to \infty} \frac{\sum_{i=1}^{n}u_{i}^{-\alpha }}{u_{\min}^{-\alpha}} = r$, $\lim_{\alpha\to \infty} \frac{\sum_{i=1}^{n}u_{i}^{-\alpha }\log(u_i)}{u_{\min}^{-\alpha}} = r\log(u_{\min})$ and therefore

\begin{equation*} \lim_{\alpha\to \infty} G(\alpha) = 0 + \frac{n}{r} r\log(u_{\min}) = n \log(u_{\min}) < 0,
\end{equation*}
where the last inequality follows since $u_{\min} < 1$, which is true by consequence of the hypothesis that not all $t_i$ are equal. It follows that, because of the intermediate value theorem, $G(\alpha )$ has at
least one root in the interval $[0,\infty )$.

We now prove that $G^{\prime }(\alpha
)$ is negative, which in turn will imply that $G(\alpha )$ cannot have more
than one root in $[0,\infty)$, where $G^{\prime }(\alpha )$ is given by 
\begin{equation}
G^{\prime }(\alpha )=-\frac{n}{\alpha ^{2}}-n\frac{\sum
\limits_{i=1}^{n}t_{i}^{-\alpha }\left( \log t_{i}\right) ^{2}\left(
\sum\limits_{i=1}^{n}t_{i}^{-\alpha }\right) -\left(
\sum\limits_{i=1}^{n}t_{i}^{-\alpha }\log t_{i}\right) ^{2}}{\left(
\sum\limits_{i=1}^{n}t_{i}^{-\alpha }\right) ^{2}}.
\end{equation}%
To show that $G^{\prime }(\alpha )$ is negative, it is enough to show that
\begin{equation}\label{equnieq2}
\sum_{i=1}^{n}t_{i}^{-\beta }\left( \log t_{i}\right) ^{2}\left(
\sum_{i=1}^{n}t_{i}^{-\beta }\right) -\left( \sum_{i=1}^{n}t_{i}^{-\beta
}\log t_{i}\right) ^{2}>0.
\end{equation}

To prove this, one can use the Cauchy-Schwartz inequality (a special case of Holder's inequality with $p=q=2$) stated as follows
\begin{equation}
\left(\sum\limits_{i=1}^{n} a_{i}^{2}\right)\left(\sum\limits_{i=1}^{n} b_{i}^{2}\right)\ge \left(\sum\limits_{i=1}^{n} a_{i} b_{i}\right)^{2},
\end{equation}%
where equality holds if $\frac{a_{i}}{b_{i}} = $ constant.

To prove the equation (\ref{equnieq2}), take $a_{i}^{2}= t_{i}^{-\alpha}(\log t_{i})^{2}$  and $b_{i}^{2}= t_{i}^{-\alpha}$. Then clearly $a_{i}b_{i}=t_{i}^{-\alpha}\log t_{i}$, and hence the inequality (\ref{equnieq2}) follows easily from the application of the Cauchy-Schwartz inequality.

It follows that $G(\alpha) = \frac{d l(\lambda(\alpha),\alpha | t)}{d \alpha}$ has only on root $\widehat{\alpha}$ such that $G(\alpha)>0$ for $\alpha < \widehat{\alpha}$ and $G(\alpha)<0$ for $\alpha>\widehat{\alpha}$ and therefore $\widehat{\alpha}$ is the only value that attains the maximum of $l(\lambda(\alpha),\alpha)$. Then
\begin{equation*} \max_{ \lambda\in \mathbb{R}^+,\ \alpha\in \mathbb{R}^+} l(\lambda, \alpha|t) = l\left(\hat{\lambda}(\hat{\alpha}),\hat{\alpha}\right) = l(\hat{\lambda},\hat{\alpha})
\end{equation*}
and $(\hat{\lambda},\hat{\alpha})$ is the only pair that attains the maximum of $l(\lambda,\alpha)$.

\end{proof}

\section{Data set}

The datasets used in Section \ref{aplicwe} is given as follow:

\begin{itemize}
\item May: 29.19, 18.47, 12.86, 151.11, 19.46, 19.46, 84.30, 19.30, 18.47, 34.12, 374.54, 19.72, 25.58, 45.74, 68.53, 36.04, 15.92, 21.89, 40.00, 44.10, 33.35, 35.49, 56.25, 24.29, 23.56, 50.85, 24.53, 13.74, 27.99, 59.27, 13.31, 41.63, 10.00, 33.62, 32.90, 27.55, 16.76, 47.00, 106.33, 21.03.

\item June: 13.64, 39.32, 10.66, 224.07, 40.90, 22.22, 14.44, 23.59, 47.02, 37.01, 432.11, 10.63, 28.51, 11.77, 25.35, 25.80, 39.73, 9.21, 22.36, 11.63, 33.35, 18.00, 18.62, 17.71, 100.10, 23.32, 11.63, 10.20, 12.04, 11.63, 50.57, 11.63, 33.72, 14.69, 12.30, 32.90, 179.75, 37.57, 7.95.

\item July: 12.98, 15.66, 13.18, 174.94, 10.35, 47.52, 13.28, 24.03, 11.40, 22.71, 43.96, 9.38, 11.40, 13.28, 14.84, 
14.44, 63.74, 12.04, 17.26, 28.74, 12.25, 10.22, 26.25, 13.31, 28.24, 12.88, 17.71, 8.82, 10.40, 7.67, 49.15, 17.93, 9.80, 105.88, 10.77, 13.49, 19.77, 34.22, 7.26.

\item August: 16.00, 9.52, 9.43, 53.72, 17.10, 8.52, 10.00, 15.23, 8.78, 28.97, 28.06, 18.26, 9.69, 51.43, 10.96, 13.74, 20.01, 10.00, 12.46, 10.40, 26.99, 7.72, 11.84, 18.39, 11.22, 13.10, 16.58, 12.46, 58.98, 7.11, 11.63, 8.24,  9.80, 15.51, 37.86, 30.20, 8.93, 14.29, 12.98, 12.01, 6.80. 

\item September: 29.19, 8.49, 7.37, 82.93, 44.18, 13.82, 22.28, 28.06, 6.84, 12.14, 153.78, 17.04, 13.47, 15.43, 30.36, 6.91, 22.12, 35.45, 44.66, 95.81, 6.18, 10.00, 58.39, 24.05, 17.03, 38.65, 47.17, 27.99, 11.84, 9.60, 6.72, 13.74, 14.60, 9.65, 10.39, 60.14, 15.51, 14.69, 16.44
\end{itemize}

\section{Code of the Metropolis-Hasting algorithm within Gibbs}

{\footnotesize
\begin{verbatim}
library(coda)
###########################################################################
### Gibbs with Metropolis-Hasting algorithm ###
### R: Iteration Number; burn: Burn in; ### 
###jump: Jump size; b= Control generation values ###
### log_posteriori: logarithm of posteriori density ###
###########################################################################

MCMC<-function(t,R,burn,jump,b=1) {
log_posteriori <- function (alfa) {
posterior= (n-2)*log(alfa)-(alfa*sum(log(t)))-n*log(sum(t^(-alfa)))
return(posterior) }  ##logarithm of the marginal posterior of alpha
valpha<-length(R+1) ; n<-length(t)
valpha[1]<-max(log(2)/(log((((2/(n*(n-1))*sum(seq(0,n-1,1)*sort(t)))-mean(t))/mean(t))+1)),1)  
##Set the initial value of alpha based on the L-moments
c1<-rep(0,times=R) 
###Starts the M-H algorithm described in Section 4.2.1
for(i in 1:R){
prop1<-rgamma(1,shape=b*valpha[i],rate=b)
ratio1<-log_posteriori(prop1)+dgamma(valpha[i],shape=b*prop1,rate=b,log=TRUE)-dgamma(
prop1,shape=b*valpha[i],rate=b,log=TRUE)-log_posteriori(valpha[i])
h<-min(1,exp(ratio1)); u1<-runif(1)
if (u1<h & is.finite(h)) {valpha[i+1]<-prop1 ; c1[i]<-0} else 
{valpha[i+1]<-valpha[i] ; c1[i]<-1} }
###Ends the M-H algorithm
valpha2<-valpha[seq(burn,R,jump)]  ###Remove the burn-in and takes jump
ge1<-abs(geweke.diag(valpha2)$z[1]) ### Compute the  Geweke diagnostic
alpha<-median(valpha2);    ### Compute the median of alpha
### Compute the median of lambda
lambda<-qgamma(0.5, shape=n, rate = sum(t^(-alpha)), lower.tail = TRUE) 
prai<-quantile(valpha2, probs = 0.025)  ## Compute the Lower credibility interval of alpha
pras<-quantile(valpha2, probs = 0.975)  ## Compute the Upper credibility interval of alpha
## Compute the Lower credibility interval of lambda
prli<-qgamma(0.025, shape=n, rate = sum(t^(-alpha)), lower.tail = TRUE)
## Compute the Upper credibility interval of lambda 
prls<-qgamma(0.975, shape=n, rate = sum(t^(-alpha)), lower.tail = TRUE)  
return(list(acep=(1-sum(c1)/length(c1)),lambda=lambda,alpha=alpha, LCI_alpha=prai,
UCI_alpha=pras, LCI_Lambda=prli, UCI_Lambda=prls, Geweke.statistics=ge1))
}

################################################################
## Example ### t: data vector ###
################################################################

rIW<-function(n,lambda,alpha) {
 U<-runif(n,0,1)
 t<-(((1/lambda)*(log(1/U)))^(-1/alpha)) 
return(t) }

set.seed(2018)
t<-rIW(n=30,lambda=4,alpha=2)

MCMC(t,R=15000,burn=500,jump=5)

$acep          ##Aceptance rate
[1] 0.2490667

$lambda        ##Posterior median of lambda
[1] 4.880441

$alpha         ##Posterior median of alpha
[1] 2.236982

$LCI_alpha     ##Lower credibility interval of alpha
    2.5% 
1.618934 

$UCI_alpha     ##Upper credibility interval of alpha
   97.5% 
2.917118 

$LCI_Lambda    ##Lower credibility interval of lambda
[1] 3.329736

$UCI_Lambda    ##Upper credibility interval of lambda
[1] 6.851465

$Geweke.statistics ## Geweke Statitics
1.835165 
\end{verbatim} }

\end{document}